\documentclass[11pt]{article}
\usepackage[utf8x]{inputenc}
\usepackage{amsmath}
\usepackage{amssymb}
\usepackage{fullpage}
\usepackage{enumerate}
\usepackage{amsthm}
\usepackage{nameref}
\newtheorem{lemma}{Lemma}[section]
\newtheorem{theorem}{Theorem}[section]
\newtheorem{corollary}{Corollary}[section]

\usepackage{indentfirst}
\usepackage{graphics}
\usepackage{complexity}
\usepackage[toc,page]{appendix}
\usepackage{verbatim}
\usepackage{authblk}
\input{Qcircuit}
\newclass{\PDQP}{PDQP}
\newclass{\DQPST}{DQP_{\mathcal{ST}}}
\newclass{\CDQP}{CDQP}

\newcommand\blank{{\mkern 2mu\cdot\mkern 2mu}}
\renewcommand{\bra}[1]{\left\langle #1 \right \vert}
\renewcommand{\ket}[1]{\left\vert #1 \right \rangle}
\newcommand{\braket}[2]{\left\langle #1 \middle \vert #2 \right\rangle}
\newcommand{\norm}[1]{\left\| #1 \right\|}
\allowdisplaybreaks
\title{The space ``just above" $\BQP$}
\author[1]{Scott Aaronson\thanks{email: aaronson@csail.mit.edu}}
\author[1]{Adam Bouland\thanks{email: adam@csail.mit.edu}}
\author[2]{Joseph Fitzsimons\thanks{email: joe.fitzsimons@nus.edu.sg}}
\author[1]{Mitchell Lee\thanks{email: mitchlee@mit.edu}}
\affil[1]{Massachusetts Institute of Technology, Cambridge, MA USA}
\affil[2]{Singapore University of Technology and Design and Centre for Quantum Technologies, National University of Singapore, Singapore}
\date{\vspace{-5ex}}
\begin{document}

\clearpage
  \maketitle
\begin{abstract}
 
 We explore the space ``just above" $\BQP$ by defining a complexity class $\PDQP$ (Product Dynamical Quantum Polynomial time) which is larger than $\BQP$ but does not contain $\NP$ relative to an oracle. The class is defined by imagining that quantum computers can perform measurements that do not collapse the wavefunction. This (non-physical) model of computation can efficiently solve problems such as Graph Isomorphism and Approximate Shortest Vector which are believed to be intractable for quantum computers. Furthermore, it can search an unstructured $N$-element list in $\tilde O(N^{1/3})$ time, but no faster than $\Omega(N^{1/4})$, and hence cannot solve $\NP$-hard problems in a black box manner. In short, this model of computation is more powerful than standard quantum computation, but only slightly so. 
 
 Our work is inspired by previous work of Aaronson on the power of sampling the histories of hidden variables. However Aaronson's work contains an error in its proof of the lower bound for search, and hence it is unclear whether or not his model allows for search in logarithmic time. Our work can be viewed as a conceptual simplification of Aaronson's approach, with a provable polynomial lower bound for search.

\end{abstract}


\section{Introduction}
\par Quantum computers are believed to be strictly more powerful than classical computers, but not so much more powerful that they can solve $\NP$-hard problems efficiently. In particular, it is known that $\BQP$, the class of languages recognizable in polynomial time by a quantum algorithm~\cite{bqp}, does not contain $\NP$ ``relative to an oracle.'' This means that there is some ``black box" problem $\mathcal{O}$ for which $\BQP^{\mathcal{O}} \not \supset \NP^{\mathcal{O}}$. (For more information about the terminology, see~\cite[pp. 72-76]{complexity}.) On the other hand, many seemingly innocuous modifications of quantum mechanics---for example, allowing nonlinear transformations \cite{nonlinear}, non-unitary transformations, postselection, or measurement statistics based on the $p$th power of the amplitudes for $p \neq 2$---increase the power of quantum computation drastically enough that they can solve $\NP$-hard problems (and even $\#\P$-hard problems) efficiently~\cite{postbqp}. As a result, it is difficult to find natural complexity classes which are bigger than $\BQP$ but which don't contain $\NP$. Quantum mechanics appears to be an ``island in theoryspace" in terms of its complexity-theoretic properties~\cite{postbqp}.

\par In this work, we explore a natural modification of quantum mechanics to obtain a complexity class $\PDQP$ which lies ``just above" $\BQP$, i.e. it contains $\BQP$, it strictly contains $\BQP$ relative to an oracle, but it still does not contain $\NP$ relative to an oracle. Our model is defined by imagining one could perform measurements which do not collapse the state, in addition to the usual projective measurements (which collapse the state). We show that quantum computers equipped with this power can solve the Graph Isomorphism problem in polynomial time, yet require $\Omega(N^{1/4})$ time to search an unordered list of $N$ elements. To our knowledge this represents the only known modification of quantum mechanics which provably does not admit polynomial time black-box algorithms for $\NP$-hard problems.

\par Our work is inspired by previous work on quantum computing with hidden variables by Aaronson~\cite{hidden}. His work defines a class $\DQP$ by imagining a hidden variable theory is true, and that an experimenter can view the evolution of the hidden variables in real time. He shows that with this power one can search in $\tilde O(N^{1/3})$ time and solve any problem in $\SZK$ in polynomial time. He additionally claims one cannot search in faster than $\Omega(N^{1/3})$ time in this model. Unfortunately, there is an error which invalidates his proof of the lower bound for search. For the interested reader, we describe this error in Appendix~\ref{errors} and correct the error for a modified version of the computational model in Appendix~\ref{app:dqpfix}. Proving the lower bound for search under Aaronson's original computational model is challenging because we have few examples of working hidden variable theories, and therefore have little understanding of how hidden variable values could correlate over time. Note, however, that an $\Omega(N^{1/3})$ lower bound for search might hold even for Aaronson's original model.

\par In this work, we jettison the machinery of hidden variable theories and instead consider the power of quantum computers which can make both usual quantum measurements and ``non-collapsing measurements." These are identical to usual quantum measurements except that they do not collapse the state. Non-collapsing measurements in some sense capture the power of hidden variable theories which are used in Aaronson's paper, while being simpler to analyze. We call the class of problems decidable in polynomial time in this model $\PDQP$, which stands for ``Product Dynamical Quantum Polynomial time"\footnote{The name comes from the fact that $\PDQP$ can be viewed as a version of $\DQP$ in which the hidden variable dynamics are governed by ``product theory", i.e. they reproduce the results of non-collapsing measurements \cite{hidden}.}. Like $\DQP$, we show that the class $\PDQP$ contains both $\SZK$ and $\BQP$, so there is an oracle $\mathcal{O}$ for which $\BQP^\mathcal{O} \neq \PDQP^\mathcal{O}$. Furthermore, we show a strong classical upper bound for $\PDQP$, namely that $\PDQP \subseteq \BPP^\PP$. In contrast the best known classical upper bound for $\DQP$ is $\EXP$ \cite{hidden}.

\par We also demonstrate that if non-collapsing measurements are possible, then there is a quantum algorithm that searches an unstructured list of $N$ elements in $\tilde O(N^{1/3})$ time, and furthermore any such algorithm takes at least $\Omega(N^{1/4})$ time. While the upper bound is simple, the proof of the lower bound uses a hybrid argument \cite{hybrid} and properties of Markov chains. We conclude that $\PDQP$ does not contain $\NP$ relative to an oracle. Therefore $\PDQP$ lies ``just above $\BQP$", while being easier to define than $\DQP$. In short, allowing non-collapsing measurements does not drastically increase the power of quantum computers, unlike many other modifications of quantum mechanics \cite{postbqp, nonlinear}.

\par Note that introducing non-collapsing measurements into quantum mechanics allows for many strange phenomena. In particular, it allows for faster-than-light communication, it allows for quantum cloning, and it renders quantum query complexity and quantum communication complexity meaningless (see Appendix ~\ref{bizarreqm} for details). As a result, we are not suggesting that ``non-collapsing measurements" should be considered seriously as an amendment to quantum theory; rather we are simply showing that they have interesting complexity-theoretic properties.

\section{Quantum computing with noncollapsing measurements}

We assume the reader is familiar with the standard definition of $\BQP$ and the basics of quantum computing; for an introduction to this topic see \cite{quantum}. We now give a formal definition of our model of quantum computing with non-collapsing measurements.

\par Let $\mathcal{Q}_P$ be an oracle that takes as input a quantum circuit $C = (U_1, M_1, U_2, M_2, \cdots, U_T, M_T)$ and an integer $\ell \geq 0$. Here each $U_i$ is a unitary operator on $\ell$ qubits composed of gates from some finite universal gate set $\mathcal{U}$, and each $M_i$ is a standard (collapsing) measurement of zero or more qubits in the computational basis. Define a (random) sequence $\{\ket{\psi_t}\}_{t=0}^{T}$ of quantum states by $\ket{\psi_0} = \ket{0}^{\otimes \ell} $ and for $t > 0$, $\ket{\psi_t}$ is the resulting (random) pure state obtained when measurement $M_t$ is applied to $U_t \ket{\psi_{t-1}}$. \ Note that we imagine the state of the system $\ket{\psi_t}$ is a (random) pure state for $0 \leq t \leq T$. \ The oracle $\mathcal{Q}_P$ samples the sequence $\{\ket{\psi_t}\}_{t=0}^{T}$ (note that the random variables $\ket{\psi_t}$ are not independent), measures $\ket{\psi_t}$ in the computational basis for every $t$ independently, and outputs the $T+1$ measurement results, which we label $v_0,v_1,\ldots v_T$, respectively. The output of $\mathcal{Q}_P$ is an element of $(\{0,1\}^\ell)^{T+1}$. Note that once the $\ket{\psi_t}$ are fixed, the $T+1$ measurement results are independent, however since the $\ket{\psi_t}$ are correlated, the measurement outcomes may be correlated.

\par $\PDQP$ (Product Dynamical Quantum Polynomial-time) is then defined as the class of all languages that can be recognized in polynomial time by a deterministic Turing machine with one query to $\mathcal{Q}_P$, with error probability at most $\frac{1}{3}$. Note that because the base machine is polynomially bounded, the circuit $C$ with which it queries $\mathcal{Q}_P$ must be polynomially sized. This class contains $\BQP$, because one can always query the oracle $\mathcal{Q}_P$ with a $\BQP$ circuit, and then ignore all output except the final measurement outcome. The constant $\frac{1}{3}$ is arbitrary: we can decrease the error probability arbitrarily close to $0$ by repetition, which can be accomplished by packing multiple copies of a quantum circuit into a single call to $\mathcal{Q}_P$. Furthermore, it turns out that the definition of $\PDQP$ is not affected by the choice of universal gate set $\mathcal{U}$; this is a consequence of the Solovay-Kitaev Theorem. See Appendix~\ref{gateset} for details.

We can think of the $T+1$ measurement samples from $\mathcal{Q}_P$ as the results of \emph{non-collapsing} measurements on the state vector, which give information about the state without changing it. For instance, let $\ket{\psi_1}=U_1 \ket{0}^{\otimes \ell}$, let $M_1, M_2, M_3$ be empty measurements, and let $U_2,U_3$ be the identity. Then the oracle $\mathcal{Q}_P$ will output the result of three independent non-collapsing measurements of $\ket{\psi_1}$ in the computational basis. The key point is that the oracle's samples do not disturb the state of the system; only the unitary operators $U_i$ and collapsing measurements $M_i$ do. The oracle $\mathcal{Q}_P$ gives us information about the intermediate stages of the quantum computation without collapsing the state; this is what gives $\PDQP$ additional power over $\BQP$.




Note that we explicitly allow for intermediate (collapsing) measurements in our model. In the definition of $\BQP$, the principle of deferred measurement tells us that this is not necessary; the power of standard quantum computers is unchanged by the inclusion of intermediate collapsing measurements. However, in our model this makes a crucial difference. Indeed, suppose that we did not allow for intermediate collapsing measurements; then this model would be simulable in $\BQP$ with a polynomial amount of overhead. If there are no intermediate measurements $M_i$, then $\ket{\psi_t}=U_tU_{t-1}\ldots U_1 \ket{0}^{\otimes \ell}$ are no longer random variables but are deterministic pure states, each preparable with a polynomially sized quantum circuit. So a $\BQP$ machine could simply prepare $\ket{\psi_1}$ and measure it, then prepare $\ket{\psi_2}$ from scratch and measure it, etc. to obtain the samples $v_0, \ldots, v_T$. This would incur at most quadratic overhead.

When we add intermediate measurements into our model, this simulation strategy no longer works. Indeed, suppose that we performed measurement $M_1$ to obtain a random state $\ket{\psi_1}$. If we wanted to reproduce this state with a $\BQP$ machine, we could try applying $M_1$ to $U_1 \ket{0}^{\otimes \ell}$. However, it might be that the probability of obtaining the same outcome for $M_1$ is exponentially small, and hence the $\BQP$ machine could not prepare another copy of $\ket{\psi_1}$ in polynomial time. 

In short, the power of this model comes from the fact that we can perform intermediate measurements which collapse the wave function, and afterwards we can examine the resulting pure state $\ket{\psi_t}$ (which might not be efficiently preparable with a $\BQP$ machine) using multiple non-collapsing measurements. In the next section we will show how to leverage these properties to solve any problem in $\SZK$ in polynomial time.

\section{$\SZK \subseteq \PDQP$}
\par We will now describe how to use the peculiarties of non-collapsing measurements to solve any problem in $\SZK$ in polynomial time. The proof uses essentially the ideas of Aaronson \cite{hidden}, with minor simplifications.

$\SZK$ was originally defined as the class of languages admitting statistical zero-knowledge proofs. The precise definition of a statistical zero-knoweledge proof can be found in~\cite{szk}, but it is not important here. $\SZK$ includes important problems such as Graph Isomorphism and Approximate Shortest Vector. It has been a long-standing open problem whether or not these problems can be solved in quantum polynomial time. Ettinger, H\o yer and Knill showed that Graph Isomorphism (and indeed any hidden subgroup problem) can be solved in a black box manner with a polynomial number of queries to the black box, but with exponential post-processing time \cite{graphisoquery}. On the other hand, Aaronson \cite{collision} showed that $\BQP$ does not admit a black-box algorithm for the collision problem, and hence there is an oracle relative to which $\SZK$ is not in $\BQP$.

\par In contrast, we show that quantum computers with non-collapsing measurements can solve any problem in $\SZK$ efficiently, i.e. $\SZK \subseteq \PDQP$. It is enough to prove that Statistical Difference, a problem shown in~\cite{szk} to be $\SZK$-complete, is in $\PDQP$. The statistical difference problem is to determine, for two functions $P_0, P_1 : \{0, 1\}^n \to \{0, 1\}^m$ specified by classical circuits, whether the distributions of $P_0(X), P_1(X)$ for uniformly random $X$ are close or far. Here, two distributions are ``close'' if their total variation distance is less than $\frac{1}{3}$ and they are ``far'' if their total variation distance is more than $\frac{2}{3}$.

We now show how to solve this efficiently if we have access to non-collapsing measurements.

\begin{theorem} The Statistical Difference problem can be solved in polynomial time in $\PDQP$, and hence $\SZK \subseteq \PDQP$.
\end{theorem}
\begin{proof}

\par By the Polarization Lemma of Sahai and Vadhan~\cite[Lemma~3.3]{szk}, we can assume that the distributions $P_0(X)$ and $P_1(X)$ have total variation distance less than $2^{-n^{c}}$ or more than $1-2^{-n^{c}}$, for any constant $c$. For now, assume that the distributions have total variation distance equal to either $1$ or $0$. 
\par Our algorithm for the statistical difference problem is as follows. Prepare the state \[\frac{1}{2^{(n+1)/2}}\sum_{b \in \{0, 1\}, x \in \{0, 1\}^n} \ket{b} \ket{x} \ket{P_b(x)}.\] Now, measure the third register with a collapsing measurement to obtain a state $\ket{\phi}$ on the first two registers. If the distributions $P_0, P_1$ have total variation distance $1$, then $\ket{\phi}$ will be of the form $\ket{b}\ket{\psi}$ for some $b$ and $\ket{\psi}$. On the other hand, if they have total variation distance $0$, then $\ket{\phi}$ will be an equal superposition $\frac{1}{\sqrt{2}}(\ket{0} \ket{\psi_0} + \ket{1} \ket{\psi_1})$ where $\ket{\psi_1}$ and $\ket{\psi_2}$ have unit norm. We can distinguish the two cases by now repeatedly performing non-collapsing measurements and examining the value of the first register. If $P_0, P_1$ have total variation distance $1$, then all of these measurements will give the same value $b$; if $P_0$ and $P_1$ have total variation distance $0$, then each of these measurements will independently give $0$ with probability $\frac{1}{2}$ and $1$ with probability $\frac{1}{2}$. We can distinguish the two cases with probability $3/4$ by performing three non-collapsing measurements and looking at whether or not they yielded identical values of the first register. 

Furthermore, the fact that the total variation distances are merely exponentially close to $0$ or $1$, rather than actually being equal to $0$ or $1$, makes little difference. One can easily show that the probability of seeing the same measurement outcome three times is at most $\frac{1}{4}+O(2^{-n^c})$ if $P_0$ and $P_1$ are exponentially close and at least $1-O(2^{-n^c})$ if $P_0$ and $P_1$ are exponentially far apart. Therefore our algorithm will have error probability at most 1/3.


\end{proof}
Hence $\SZK$ is in $\PDQP$, and furthermore we can solve $\SZK$ problems in $\PDQP$ in a black box manner, i.e. relative to any oracle. Since~\cite{collision} has the result that $\SZK \not \subset \BQP$ relative to an oracle, we have the immediate corollary\footnote{Note that when we say $\PDQP^{\mathcal{O}}$, we mean that circuits given in the input to $\mathcal{Q}_P$ in the definition of $\PDQP$ can contain quantum calls to the oracle.}:

\begin{corollary} There exists an oracle $\mathcal{O}$ such that $\PDQP^\mathcal{O} \neq \BQP^\mathcal{O}$.
\end{corollary}

\section{Search in $\tilde O(N^{1/3})$ time}\label{sec:search13}
\par Suppose that we are given query access to a function $f:\{0, 1\}^n \to \{0, 1\}$ such that the preimage $f^{-1}(1)$ contains exactly one element, $x$. In the classical randomized computational model, we can find $x$ in $O(N)$ time, where $N = 2^{n}$, but no faster. In the quantum computational model, on the other hand, we can find $x$ in $O(N^{1/2})$ time using Grover's search algorithm~\cite{grover}, but no faster \cite{hybrid}. 

Here we show that quantum computers equipped with non-collapsing measurements can search in $\tilde O (N^{1/3})$ time, where the tilde hides factors in $\log N$. The basic idea is to run $N^{1/3}$ Grover iterations, and then make $N^{1/3}$ non-collapsing measurements of the resulting state. Then with high probability the the marked item will be seen. This is a simplification of the proof given in ~\cite[Theorem~10]{hidden} for $\DQP$. We now formalize this idea below:

\begin{theorem}
 Suppose, in the definition of $\PDQP$, that the unitary operators $U_1, \cdots, U_T$ are now allowed to query $f$. That is, we are given access to the $n$-qubit gate $U_f$ defined by $U_f \ket{y} = (-1)^{f(y)} \ket{y}$ for all $y \in \{0,1\}^n$, as well as controlled-$U_f$. Then there is a $\PDQP$ algorithm to find the value of $x$ that uses $O(N^{1/3})$ queries and $\tilde O(N^{1/3})$ time.
\end{theorem}

\begin{proof}
Prepare the uniform superposition of all basis states, apply $N^{1/3}$ Grover iterations~\cite{grover}, then query the oracle to record whether or not each basis state is marked in an ancilla. We obtain the state $\alpha \ket{x}\ket{1} + \beta \sum_{y \in \{0, 1\}^n} \ket{y}\ket{0}$ with
\begin{align*}
\alpha &= \frac{1}{\sqrt{2^{n/3} + 2^{-n/3 + 1} + 1}} & & \beta = 2^{-n/3} \alpha.
\end{align*}
Now make $O(N^{1/3} \log N)$ non-collapsing measurements. We claim that with high probability, the marked item $x$ will appear at least once. Indeed, the marked item $x$ appears with probability at least $\Omega\left( N^{-1/3}\right)$ in each non-collapsing measurement outcome, so it occurs at least once with probability more than $1 - (\log N + 1)e^{-\log N} = 1 - o(1)$. 
\end{proof}
\par Note that if we are willing to use an enormous amount of \emph{time}, we can search in the $\PDQP$ model using only one \emph{query}: just query the oracle in superposition and then perform $O(N)$  non-collapsing measurements. Indeed as we note in the introduction, any function $f$ has query complexity $1$ in this model, although this approach requires exponentially many non-collapsing measurements. Therefore in this model of computation, the relevant measure of complexity of an algorithm is the number of queries $Q$ plus the number of non-collapsing measurements $T$ used by the algorithm. Our above algorithm uses $Q+T=\tilde O(N^{1/3})$ of each, with $O(N^{1/3})$ post-processing time, so we say it ``runs in time $\tilde O(N^{1/3})$".

\section{Lower bounds for search}

We now show that our search algorithm in section~\ref{sec:search13} cannot be improved by much; in particular there is no way to solve search in faster than $N^{1/4}$ time, even with non-collapsing measurements.

\begin{theorem} \label{lowerbound}
 Suppose, in the definition of $\PDQP$, that the unitary operators $U_1, \cdots, U_T$ are now allowed to query $f$. Let $Q$ be the number of queries to $f$ made by a $\PDQP$ algorithm, and $T$ be the number of non-collapsing measurements. Then any $\PDQP$ algorithm to find the value of $x$  obeys  $Q+T=\Omega(N^{1/4})$, and hence search requires $\Omega(N^{1/4})$ time.
\end{theorem}

\par In other words, there is no ``black box" polynomial-time algorithm for $\NP$-hard problems, even when given access to non-collapsing measurements. This is evidence that the class $\PDQP$ does not contain $\NP$. The following corollary follows immediately from the well-known ``diagonalization method'' of Baker, Gill, and Solovay~\cite{relative}:

\begin{corollary} There exists an oracle $\mathcal{O}$ such that $\NP^\mathcal{O} \not \subset \PDQP^\mathcal{O}$.
\end{corollary}

We now proceed to a proof of Theorem \ref{lowerbound}. The following lemma is essential: it bounds the total variation distance between two Markov distributions.
\begin{lemma}
\label{union}
 Suppose that $T \geq 1$, and that $v = (v_0, \cdots, v_T)$ is a random variable governed by a Markov distribution. That is, for all $1 \leq i \leq T$, $v_i$ is independent of $v_0, \cdots, v_{i-2}$ conditioned on a particular value of $v_{i-1}$. Let $w = (w_0, \cdots, w_T)$ be another random variable governed by a Markov distribution. If $d_{TV}(\blank, \blank)$ denotes the total variation distance between random variables, then \[d_{TV}(v, w) \leq 2 \sum_{i=1}^{T} d_{TV}((v_{i-1}, v_i), (w_{i-1}, w_i)).\]
\end{lemma}
\begin{proof}
 \par We proceed by induction on $T$. The base case $T=1$ is trivial. For $T > 1$, since $w_T$ depends only on $w_{T-1}$ (by the Markov property), it is equal to $A(w_{T-1})$ for some randomized process $A$; let $w'_T := A(v_{T-1})$ be a variable that depends on $v_{T-1}$ in exactly the same way that $w_T$ depends on $w_{T-1}$. Then, define the random variable $v' = (v_0, \cdots, v_{T-1}, w'_T)$. By the triangle inequality,
\begin{equation} \label{triangle1}
  d_{TV}(v, w) \leq d_{TV}(v,v') + d_{TV}(v',w).
\end{equation}
 \par Applying the same randomized process to two random variables cannot increase their total variation distance~\cite{szk}. We can generate random variables identically distributed to $v$ and $v'$ by applying a suitable randomized process to $(v_{T-1}, v_T)$ and $(v_{T-1}, w'_T)$. We can also generate random variables identically distributed to $v'$ and $w$ by applying a suitable randomized process to $(v_0, \cdots, v_{T-1})$ and $(w_0, \cdots, w_{T-1})$. Therefore, the right hand side of (\ref{triangle1}) is bounded above by \[d_{TV}((v_{T-1}, v_T), (v_{T-1}, w'_T)) + d_{TV}((v_0, \cdots, v_{T-1}), (w_0, \cdots, w_{T-1})).\] By the triangle inequality,
 \begin{align*}
 d_{TV}((v_{T-1}, v_T), (v_{T-1}, w'_T)) &\leq d_{TV}((v_{T-1}, v_T), (w_{T-1}, w_T)) + d_{TV}((w_{T-1}, w_T), (v_{T-1}, w'_T)) \\
 &= d_{TV}((v_{T-1}, v_T), (w_{T-1}, w_T)) + d_{TV}(v_{T-1}, w_{T-1}) \\
 &\leq 2 d_{TV}((v_{T-1}, v_T), (w_{T-1}, w_T)).
 \end{align*}
 Putting all of this together, \[d_{TV}(v,w) \leq 2 d_{TV}((v_{T-1}, v_T), (w_{T-1}, w_T)) + d_{TV}((v_0, \cdots, v_{T-1}), (w_0, \cdots, w_{T-1})).\] The result follows from induction.
\end{proof}
\begin{lemma}\label{trace}
 The trace distance between two pure states $\ket{\psi} \bra{\psi}$ and $\ket{\phi} \bra{\phi}$ is less than or equal to the 2-norm $\norm{\ket{\psi} - \ket{\phi}}_2$.
\end{lemma}
\begin{proof}
 The trace distance between $\ket{\psi} \bra{\psi}$ and $\ket{\phi} \bra{\phi}$ is equal to $\sqrt{1 - |\braket{\psi}{\phi}|^2}$~\cite[p. 415]{quantum}, and the 2-norm $\norm{\ket{\psi} - \ket{\phi}}_2$ is $\sqrt{2 - 2\mathrm{Re}(\braket{\psi}{\phi})}$. The inequality follows from $|\braket{\psi}{\phi}| \leq 1$.
\end{proof}
From the hybrid argument of \cite{hybrid}, we have the following:
\begin{lemma}
\label{hybrid}
 For all $t$, if there are no measurements made before time $t$, $\sum_{x=0}^{N-1} \norm{\ket{\psi_t} - \ket{\psi_t(x)}}_2^2 \leq 4Q^2.$ \end{lemma}
\begin{proof}[Proof of Theorem~\ref{lowerbound}]
\par Since it is always possible to copy measured qubits, we can assume that qubits which are measured in an intermediate step of the algorithm are never directly modified again. Now, assume that the algorithm uses $\ell$ qubits and applies unitary operators $U_1, \cdots, U_T$, each of which is either a (controlled) query to the search function $f$ or a gate from the finite universal gate set $\mathcal{U}$. The measurements $M_1 \ldots M_T$ (which may or may not be empty) are applied between the operators $U_1 \ldots U_T$.
\par Let $v(x) = (v_0(x), v_1(x), \cdots, v_T(x))$ be the non-collapsing measurement results when the marked item is $x$, so that $v_i(x)$ is sampled immediately before the application of $U_{i+1}$. Let $v = (v_0, \cdots, v_T)$ be the non-collapsing measurement results when there is no marked item. In general, both $v(x)$ and $v$ are random variables. Since the postprocessing step can distinguish the distributions of $v$ and $v(x)$ with success probability $2/3$, $d_{TV}(v, v(x)) \geq \frac{1}{3}$ for all $x$. On the other hand, each $v$ and $v(x)$ is a Markov process. Therefore, by Lemma~\ref{union}, \[d_{TV}(v, v(x)) \leq 2 \sum_{i=1}^T d_{TV}((v_{i-1}, v_i), (v_{i-1}(x), v_i(x))).\]
\par Now, we bound the term $d_{x, i} := d_{TV}((v_{i-1}, v_i), (v_{i-1}(x), v_i(x)))$. Since it is possible to defer measurements in a quantum circuit to a later stage~\cite[p. 186]{quantum}, we can assume that all intermediate measurements that occurred before the application of $U_{i}$ occurred immediately before the sampling of $v_i$. Suppose that these measurements were applied to the first $k$ qubits of the state. Let $\ket{\phi}$ and $\ket{\phi(x)}$ be the state vectors immediately before these measurements. Then, we decompose $\ket{\phi} = \sum_{s \in \{0, 1\}^k} \alpha_s \ket{s} \ket{\phi_s}$ and $\ket{\phi(x)} = \sum_{s \in \{0, 1\}^k} \beta_s \ket{s} \ket{\phi_s(x)}$. Possible values for $(v_{i-1}, v_i)$ and $(v_{i-1}(x), v_i(x))$ can be written in the form $(st_1, st_2)$, where $s$ is a $k$-bit string and $t_1, t_2$ are $(\ell-k)$-bit strings.
\par Assume for now that $U_i$ does not contain a query to $f$. Then, since it does not affect the first $k$ qubits, it can be decomposed into the sum $\sum_{s \in \{0, 1\}^k} \ket{s} V_s \bra{s}$ for some unitary operators $V_s$. The transformation $U_i$ can be thought of as applying the unitary $V_s$ to the last $\ell - k$ qubits if the (measured) first $k$ qubits are equal to $s$. Then, the probability that $(v_{i-1}, v_i) = (st_1, st_2)$ is equal to $|\alpha_s|^2 |\braket{t_1}{\phi_s}|^2 |\bra{t_2}V_s\ket{\phi_s}|^2,$ and the probability that $(v_{i-1}(x), v_i(x)) = (st_1, st_2)$ is equal to $|\beta_s|^2 |\braket{t_1}{\phi_s(x)}|^2 |\bra{t_2}V_s\ket{\phi_s(x)}|^2.$ Therefore, the total variation distance $d_{x,i}$ is by the triangle inequality
\begin{align*}
d_{x,i} &= \frac{1}{2} \sum_{s, t_1, t_2} \left||\alpha_s|^2 |\braket{t_1}{\phi_s}|^2 |\bra{t_2}V_s\ket{\phi_s}|^2 - |\beta_s|^2 |\braket{t_1}{\phi_s(x)}|^2 |\bra{t_2}V_s\ket{\phi_s(x)}|^2 \right| \\
&\leq \frac{1}{2} \sum_{s, t_1, t_2} \left(\left||\alpha_s|^2 |\braket{t_1}{\phi_s(x)}|^2 |\bra{t_2}V_s\ket{\phi_s(x)}|^2 -  |\beta_s|^2 |\braket{t_1}{\phi_s(x)}|^2 |\bra{t_2}V_s\ket{\phi_s(x)}|^2\right|  \right) \\
&\qquad + \frac{1}{2} \sum_{s, t_1, t_2}\left(|\alpha_s|^2 \left||\braket{t_1}{\phi_s}|^2 |\bra{t_2}V_s\ket{\phi_s}|^2 - |\braket{t_1}{\phi_s(x)}|^2 |\bra{t_2}V_s\ket{\phi_s}|^2 \right| \right) \\
&\qquad + \frac{1}{2} \sum_{s, t_1, t_2}\left(|\alpha_s|^2 \left||\braket{t_1}{\phi_s(x)}|^2 |\bra{t_2}V_s\ket{\phi_s}|^2 - |\braket{t_1}{\phi_s(x)}|^2 |\bra{t_2}V_s\ket{\phi_s(x)}|^2 \right| \right) \\
&=: \frac{1}{2}(S_1 + S_2 + S_3)
\end{align*}
where $S_1, S_2, S_3$ are the three sums written above, which range over $s \in \{0,1\}^k$ and $t_1, t_2 \in \{0,1\}^{\ell - k}$. Now, we have:
\begin{align*}
 S_1 &:= \sum_{s, t_1, t_2} \left(\left||\alpha_s|^2 |\braket{t_1}{\phi_s(x)}|^2 |\bra{t_2}V_s\ket{\phi_s(x)}|^2 -  |\beta_s|^2 |\braket{t_1}{\phi_s(x)}|^2 |\bra{t_2}V_s\ket{\phi_s(x)}|^2\right|  \right) \\
  &= \sum_s \left||\alpha_s|^2 - |\beta_s|^2\right| \left(\sum_{t_1, t_2} |\braket{t_1}{\phi_s(x)}|^2 |\bra{t_2}V_s\ket{\phi_s(x)}|^2\right) \\
 &= \sum_s \left||\alpha_s|^2 - |\beta_s|^2\right| \\
  &\leq \norm{\ket{\phi} \bra{\phi} - \ket{\phi(x)} \bra{\phi(x)}}_{tr} \\
  &\leq 2\norm{\ket{\phi(x)} - \ket{\phi}}_2.
\end{align*}
Additionally,
\begin{align*}
 S_2 &:= \sum_{s, t_1, t_2}\left(|\alpha_s|^2 \left||\braket{t_1}{\phi_s}|^2 |\bra{t_2}V_s\ket{\phi_s}|^2 - |\braket{t_1}{\phi_s(x)}|^2 |\bra{t_2}V_s\ket{\phi_s}|^2 \right| \right) \\
  &= \sum_{s, t_1} \left(|\alpha_s|^2 \left||\braket{t_1}{\phi_s}|^2 - |\braket{t_1}{\phi_s(x)}|^2 \right| \right) \\
  &\leq \sum_{s, t_1} \left(\left| |\alpha_s|^2 |\braket{t_1}{\phi_s}|^2 - |\beta_s|^2 |\braket{t_1}{\phi_s(x)}|^2 \right| \right) + \sum_{s, t_1}\left(\left||\alpha_s|^2 |\braket{t_1}{\phi_s(x)}|^2 - |\beta_s|^2 |\braket{t_1}{\phi_s(x)}|^2 \right|\right) \\
  &= \sum_{s, t_1} \left(\left| |\alpha_s|^2 |\braket{t_1}{\phi_s}|^2 - |\beta_s|^2 |\braket{t_1}{\phi_s(x)}|^2 \right| \right) + \sum_{s}\left(\left||\alpha_s|^2 - |\beta_s|^2 \right|\right) \\
  &\leq 2\norm{\ket{\phi} \bra{\phi} - \ket{\phi(x)} \bra{\phi(x)}}_{tr} \\
  &\leq 4\norm{\ket{\phi(x)} - \ket{\phi}}_2.
\end{align*}
Finally, 
\begin{align*}
S_3 &= \sum_{s, t_1, t_2}\left(|\alpha_s|^2 \left||\braket{t_1}{\phi_s(x)}|^2 |\bra{t_2}V_s\ket{\phi_s}|^2 - |\braket{t_1}{\phi_s(x)}|^2 |\bra{t_2}V_s\ket{\phi_s(x)}|^2 \right| \right) \\
&= \sum_{s, t_2} \left(|\alpha_s|^2 \left||\bra{t_2}V_s\ket{\phi_s}|^2 - |\bra{t_2}V_s\ket{\phi_s(x)}|^2 \right| \right) \\
&\leq \sum_{s, t_2} \left(\left||\alpha_s|^2 |\bra{t_2}V_s\ket{\phi_s}|^2 - |\beta_s|^2 |\bra{t_2}V_s\ket{\phi_s(x)}|^2 \right| \right) \\
&\qquad + \sum_{s, t_2} \left(\left||\alpha_s|^2 |\bra{t_2}V_s\ket{\phi_s(x)}|^2 - |\beta_s|^2 |\bra{t_2}V_s\ket{\phi_s(x)}|^2 \right| \right) \\
&= \sum_{s, t_2} \left(\left||\alpha_s|^2 |\bra{t_2}V_s\ket{\phi_s}|^2 - |\beta_s|^2 |\bra{t_2}V_s\ket{\phi_s(x)}|^2 \right| \right) + \sum_{s} \left(\left||\alpha_s|^2 - |\beta_s|^2 \right| \right) \\
&\leq 2\norm{\ket{\phi} \bra{\phi} - \ket{\phi(x)} \bra{\phi(x)}}_{tr} \\
&= 4\norm{\ket{\phi(x)} - \ket{\phi}}_2
\end{align*}
Therefore, \[d_{x,i} \leq \frac{1}{2}(S_1 + S_2 + S_3) \leq 5 \norm{\ket{\phi(x)} - \ket{\phi}}_2.\]
\par On the other hand, if $U_i$ is a query to $f$, then it only applies a local phase of $-1$ to some of the probability amplitudes of $\ket{\phi}$ and $\ket{\phi_x}$. Therefore, the same argument still shows that $d_{x,i} \leq 5 \norm{\ket{\phi(x)} - \ket{\phi}}_2$.
\par By the Cauchy-Schwarz inequality and Lemma~\ref{hybrid},
\begin{align*}
 \frac{1}{N}\sum_{x=0}^{N-1} d_{x, i} &\leq 5 \cdot \frac{1}{N}\sum_{x=0}^{N-1} \norm{\ket{\phi(x)} - \ket{\phi}}_2 \\
 &\leq 5 \sqrt{\frac{1}{N}\sum_{x=0}^{N-1} \norm{\ket{\phi(x)} - \ket{\phi}}^2_2} \\
 &\leq \frac{10Q}{\sqrt{N}}
\end{align*}
for all $i$.  Therefore, there is some $x$ for which
\[
 d_{TV}(v, v(x)) \leq 2 \sum_{i=1}^T d_{x,i}
 \leq \frac{20TQ}{\sqrt{N}}.
\]
On the other hand, $d_{TV}(v, v(x)) \geq \frac{1}{3}$ for all $x$, so \[\frac{20TQ}{\sqrt{N}} \geq \frac{1}{3},\] and the running time of the algorithm is at least $T + Q = \Omega(N^{1/4})$.
\end{proof}

\section{An upper bound on $\PDQP$}

We now show that $\PDQP$ is contained in the class $\BPP^\PP$. This places our class in the second level of the counting hierarchy. By comparison, the best known upper bound for $\BQP$ is $\PP$ \cite{bqppp}.

\begin{theorem} $\PDQP \subseteq \BPP^\PP$. \label{thm:upperbound}

\end{theorem}

\begin{proof}

First note that $\BPP^\PP = \BPP^{\#\P}$, because one can always use a $\PP$ oracle to count with only polynomial overhead. Therefore it suffices to show $\PDQP \subseteq \BPP^{\#\P}$. We now show how to simulate the sampling oracle $\mathcal{Q}_P$ in $\BPP^{\#P}$; since $\PDQP=\BPP^{\mathcal{Q}_P,1}$, this implies the claim.

Suppose we wish to simulate a sample from the oracle $\mathcal{Q}_P$ with input circuit $C=(U_1,M_1,\ldots U_T,M_T)$ on $n$ qubits. Since the choice of gate set does not matter (Appendix ~\ref{gateset}), without loss of generality we can assume our circuit is composed of only Toffoli and Hadamard gates, which are universal by a result of Shi \cite{Shi2003}.

We first simulate the result of the measurement $M_1$. Suppose without loss of generality that $M_1$ measures the first $k$ qubits and gets outcome $x_1 \ldots x_k \in \{0,1\}^k$. Following the techniques of Adleman, DeMarrais, and Huang \cite{bqppp}, we can write the probability that $x_1$ is 0 or 1 as an exponential sum of poly-time-computable terms (since $U_1$ is specified by a poly-sized circuit). Since we chose Hadamard and Toffoli as our gate set, all terms in the sum are of the form $\frac{\pm 1}{2^{k}}$, where $k$ is the number of Hadamard gates in $U_1$. Hence using the $\#P$ oracle, we can compute $\mathrm{Pr}[x_1=1]$ exactly in binary, and then flip a coin with bias $p$ using the base $\BPP$ machine to obtain outcome $x_1 \in {0,1}$ with this probability. 

We've now sampled the value of $x_1$. To sample the value of $x_2$, note that we can also express $\Pr[x_2=1 | x_1=0]$ as a sum of exponentially many terms, each of which is poly-time computable and takes values in $\frac{\pm 1}{2^k}$. Therefore using the $\#\P$ oracle, we can exactly compute the \emph{conditional} probability that $x_2=1$ given our sampled value of $x_1$; in other words the $\#\P$ oracle can compute the probabilities of measurement outcomes under post-selection. In this way we can sample $x_2$, then $x_3$, etc. obtain a sample $x_1\ldots x_k \in \{0,1\}^k$ as desired.

Now suppose we wish to sample the variable $v_1 \in \{0,1\}^n$ which is the result of a hidden measurement on the state remaining after measurement $M_1$ yields value $x_1\ldots x_k$. As noted above, using the $\#\P$ oracle, we can compute the marginal probability that any qubit is 1, postselected on a particular measurement outcome. Hence using the $\#\P$ oracle, we can draw the sample $v_1$ using $n$ queries to the oracle. We can continue this process to simulate $M_2$, then sample $v_2$, etc. Therefore l we can draw a sample from $\mathcal{Q}_P$ using $O(nT)$ queries to the $\#\P$ oracle.

\end{proof}

An open question is whether or not we can improve this upper bound to show $\PDQP \subseteq \PP$. One promising approach to doing so is to use the fact that $\PP=\PostBQP$ \cite{postbqp}, and design a post-selected quantum circuit to simulate the oracle $\mathcal{Q}_P$. However, the most naive way of trying to do this fails. Suppose that one tried the following: to simulate the oracle's output under $C=(U_1,M_1, \ldots U_T, M_T)$ on $n$ qubits, create a post-selected circuit $C'$ on $nT$ qubits which runs $U_1M_1$ on the first $n$ qubits, $U_1M_1U_2M_2$ on the second $n$ qubits, etc, and post-selects on them receiving the same outcomes for the intermediate measurements. While this superficially looks like what the oracle $\mathcal{Q}_P$ performs, this approach does not sample from the correct distribution on outputs. Suppose the probability that the outcome of $M_1$ is $1$ is $p$. Then the probability one sees $M_1=1$ in the final output of $C'$ will be $\frac{p^T}{p^T + (1-p)^T}$, while the quantum oracle $\mathcal{Q}_P$ will sample $M_1=1$ with probability $p$. For this reason it seems difficult to generate a sample from $\mathcal{Q}_P$ with a post-selected circuit, and hence difficult to place $\PDQP$ in $\PP$.

\section{Open questions for further research}
\par We leave many questions about the complexity classes $\DQP$ and $\PDQP$ unanswered.
\begin{enumerate}
 
 \item We demonstrated a $\tilde O(N^{1/3})$-time algorithm for the search problem in the $\PDQP$ model, as well as the result that any search algorithm takes $\Omega(N^{1/4})$ time. Is it possible to close the gap between these two bounds? If we disallow intermediate collapsing measurements, then we can prove an $N^{1/3}$ lower bound for search (a proof is included in Appendix~\ref{app:search13}). However proving an $N^{1/3}$ lower bound when there are intermediate measurements remains open.
 \item Can we demonstrate a lower bound, superpolynomial in $\log N$, for the running time of a search algorithm in the $\DQP$ model? The proof given in~\cite{hidden} of an $\Omega(N^{1/3})$ lower bound is flawed (as discussed in Appendix~\ref{errors}).
 \item Is there a hierarchy of computational models for which the $k$th allows searching in $\tilde O(N^{1/k})$ time?
 \item Can we improve the upper bound $\PDQP \subseteq \BPP^\PP$ to $\PDQP \subseteq \PP$? One possible way to approach this problem is to use the alternative formulation of $\PP$ as $\PostBQP$~\cite{postbqp}, however a straightforward application of this result does not seem to work.

\item  What is the power of quantum computers which have the ability to clone quantum states? Such devices could clearly simulate computations in $\PDQP$ - to simulate a non-collapsing measurement, simply clone the state and measure in the computational basis - but may be more powerful than $\PDQP$.


\end{enumerate}

\section{Acknowledgements}
S.A. was supported in part by an Alan T. Waterman Award. A.B. was supported in part by the National Science Foundation Graduate Research Fellowship under Grant No. 1122374 and by the Center for Science of Information (CSoI), an NSF Science and Technology Center, under grant agreement CCF-0939370. J.F. was supported in part by the Singapore National Research Foundation under NRF Award No. NRF-NRFF2013-01. M.L. was supported by the MIT SPUR program.

\begin{appendices}

\section{The error in the $\DQP$ search time lower bound, and a roadmap for correcting it} \label{errors}
\par

We now describe the error in Aaronson's original proof of an $\Omega(N^{1/3})$ lower bound for search in the $\DQP$ model, which is related to the fact that hidden variable theories can have strong correlations between their values at different times.

\subsection{The class $\DQP$}
\par We first describe the formal definition of the complexity class $\DQP$, which is based on the notion of a hidden-variable theory. A hidden-variable theory is an interpretation of quantum mechanics in which a quantum system is described by both a state vector and a definite state (called the ``hidden variable''), which determines the result of measurements on the system. When a transformation is applied to the system, the state vector evolves by a unitary linear transformation, like in ordinary quantum mechanics, and the hidden variable evolves stochastically according to the state vector and the unitary linear transformation. According to the Kochen-Specker theorem~\cite{kochenspecker}, it is impossible for the hidden variable to determine a result for all possible measurements on the system. Therefore, in what follows, we will only ever measure the quantum system in some fixed basis.
\par Suppose that our quantum system is described by a Hilbert space with $N$ basis states $\ket{1}, \cdots, \ket{N}$. Then, the hidden variable has one of the values $1, \cdots, N$. The hidden-variable theory specifies the probabilities that the hidden variable changes from $i$ to $j$ given that the state was $\ket{\psi}$ and was transformed by the unitary $U$. More precisely, a hidden variable theory $\mathcal{T}$ is specified by a stochastic matrix $S_{\mathcal{T}}(\ket{\psi}, U)$ for every state $\ket{\psi}$ and unitary transformation $U$ of dimension $N$, which indicates how the hidden variable evolves when the state transforms from $\ket{\psi}$ to $U\ket{\psi}$. If $\mathcal{T}$ is understood from context, then we simply write $S(\ket{\psi}, U)$. Suppose $\ket{\psi}=\Sigma_i \alpha_i \ket{i}$ and $U\ket{\psi} = \Sigma_j \beta_j \ket{j}$. The hidden-variable theory must be consistent with the predictions of quantum mechanics, which is to say that the probability that the hidden variable is equal to $i$ is equal to $|\alpha_{i}|^2$. This means that the stochastic matrix $S = S(\ket{\psi}, U)$ must satisfy \[|\beta_{j}|^2 = \sum_{i=1}^n |\alpha_i|^2 (S)_{ij}.\] Other ``reasonable'' properties that we might expect a hidden-variable theory to have, for example that \[S(\ket{\psi}, WV) = S(\ket{\psi}, V) S(V\ket{\psi}, W),\] need not be satisfied.
\par Sometimes, the hidden-variable theory is described instead by the matrix $P = P(\ket{\psi}, U)$ of joint probabilities, defined by $(P)_{ij} = |\alpha_i|^2(S)_{ij}$. The matrix $S$ is then recovered by \[S(\ket{\psi}, U) = \lim_{\epsilon \to 0^+} \frac{(P(\ket{\psi_\epsilon}, U))_{ij}}{|(\ket{\psi_\epsilon})_{i}|^2}\] where $\ket{\psi_\epsilon} = \sqrt{1-\epsilon}\ket{\psi} + \sqrt{\epsilon} \frac{1}{2^{N/2}}\Sigma_i\ket{i}$. The function $P(\ket{\psi}, U)$ only defines a hidden-variable theory if this limit actually exists.
\par The hidden-variable theory is called \textit{local} if unitary transformations on some subsystem $A$ of the system do not affect the value of the hidden variable on a separate subsystem $B$. A stronger property is \textit{indifference}, which is the property that if $U$ is block-diagonal, then $S(\ket{\psi}, U)$ is block-diagonal with the same block structure or some refinement thereof. It is called \textit{commutative} if the order of unitaries applied to separate subsystems is irrelevant. A theorem of Bell states that no hidden-variable theory satisfies both locality and commutativity. The theory is called \textit{robust} if for every polynomial $q(N)$, there is a polynomial $p(N)$ such that perturbing the unitary $U$ and density matrix $\ket{\psi}$ by at most $\frac{1}{p(N)}$ in the infinity norm changes the matrix $P(\ket{\psi}, U)$ by at most $\frac{1}{q(N)}$ in the infinity norm. An example of a robust indifferent hidden variable theory is the flow theory $\mathcal{FT}$ defined in~\cite{hidden}, which is based on 
network flows. For a more detailed treatment of hidden variable theories, see \cite{hidden}.

\par The complexity class $\DQP$ (Dynamical Quantum Polynomial Time) is the class of all problems solvable efficiently in the dynamic quantum model of computation. The basic idea is that a dynamic quantum algorithm is allowed to see the whole history of a hidden variable through some quantum computation (and postprocess it classically), as opposed to a quantum algorithm which can only see the final value of the hidden variable. 
\par More formally, suppose that $U_1, \cdots, U_T$ are unitary transformations on $\ell$ qubits, each specified by a sequence of gates from some finite universal gate set $\mathcal{U}$. Then, a \textit{history} of the hidden variable is a sequence $(v_0, \cdots, v_T)$ of computational basis states, with $v_0 = \ket{0}^{\otimes \ell}$. For any hidden-variable theory $\mathcal{T}$, the rule \[\Pr[v = (v_0, \cdots, v_T)] = \prod_{k=0}^{T-1} (S_{\mathcal{T}}(U_{k} \cdots U_1 \ket{0}^{\otimes \ell}, U_{k+1}))_{v_k v_{k+1}}\] defines a Markov distribution on histories. The oracle $\mathcal{O}(\mathcal{T})$ takes as input the unitaries $(U_1, \cdots, U_T)$, specified by sequences of gates from $\mathcal{U}$, and outputs a sample from this distribution.
\par Now, we are ready to define the complexity class $\DQP$. The computational model is a deterministic classical polynomial-time Turing machine $A$ that is allowed one oracle query to $\mathcal{O}(\mathcal{T})$. A language $L$ is in $\DQP$ if there is such a Turing machine $A$, such that for \emph{any} robust indifferent hidden-variable theory $\mathcal{T}$, the machine $A$ correctly decides, with probability at least $2/3$, whether a string of length $n$ is in $L$, for all sufficiently large $n$. It follows from the principle of deferred measurement that $\DQP \supset \BQP$, because viewing the entire history of a quantum system is at least as powerful as observing it only at the end of a computation~\cite{hidden}. It is important that there is one machine $A$ that works for all robust indifferent hidden-variable theories $\mathcal{T}$.

\subsection{The error}

We now describe the error in Aaronson's proof that any algorithm for the search problem in $\DQP$ takes at least $\Omega(N^{1/3})$ time. His proof is based on the hybrid argument: it shows that changing the marked item from $x$ to $x^*$ does not affect the distribution of any particular entry $v_i$ of the hidden-variable history by very much (in the total variation distance). This part of the proof is correct. However, from there he claims that this implies the total variation distance between the entire hidden variable histories $v, w$ is small, using the following inequality

\[d_{TV}(v, w) \leq  \sum_{i=0}^{T} d_{TV}(v_i, w_i).\]

While this inequality looks quite similar to Lemma~\ref{union} of our paper, it is false. The reason is that correlations between the $v_i$'s in a Markov chain can cause the total variation distance between the Markov chains to be high, while the total variation distance between the marginals is small. A specific counterexample is $T = 1$, where $v$ is $(0, 0)$ with probability $\frac{1}{2}$ and $(1,1)$ with probabilty $\frac{1}{2}$, and $w$ is $(0, 1)$ with probability $\frac{1}{2}$ and $(1,0)$ with probabilty $\frac{1}{2}$. These distributions are perfectly distinguishable, but they have the property that their marginals on any entry are identical (a 50-50 coin flip). Hence \[d_{TV}(v, w) = 1\] for this distribution whereas \[\sum_{i=0}^{T} d_{TV}(v_i, w_i) = 0\]  

Although $d_{TV}(v,w)$ cannot be upper bounded in this way, this sort of argument does show that for some item location, the probability of seeing the marked item in the hidden variable history is upper bounded by $O\left(\frac{Q^2T}{N}\right)$ (this follows from the hybrid argument and the union bound). So any search algorithm in $\DQP$ which is required to see the marked item takes at least $\Omega(N^{1/3})$ time.  However, it is possible that a $\DQP$ algorithm could infer the marked item's presence by observing correlations in the hidden variable history, without ever seeing the marked item itself. This possibility is what breaks the proof.

In order to fix this step in Aaronson's proof, one would have to show that $d_{TV}((v_{i-1}, v_i), (w_{i-1}, w_i))$ is small for each $i$, and then apply Lemma ~\ref{union} of our paper to bound the total variation distance between $v$ and $w$. Furthermore, since a $\DQP$ algorithm is required to work for all indifferent or robust hidden variable theories, one would only need to exhibit a single hidden variable theory in which is this is small. However, we only know of one indifferent and robust hidden varible theory (``flow theory"), and it remains open whether or not it satisfies this property.

\subsection{A proposed roadmap for fixing the error}

One way to fix this lower bound would be to find a hidden variable theory which is extremely robust to small perturbations. By the hybrid argument, we know that for any search algorithm making few queries, there will exist a marked item $x$ for which the state of the system $\ket{\psi^x}$ with the item $x$ present is $\epsilon$-close (where $\epsilon \approxeq \frac{Q}{\sqrt{N}}$) to the state $\ket{\psi}$ without the marked item. 

Call a hidden variable theory \emph{strongly robust} if, for all states $\psi, \phi$ that are $\epsilon$-close, and all $U,U'$ that are $\epsilon$-close,  
\[|P(\psi,U)-P(\phi, U')|_1 \leq \poly (\epsilon) \polylog(N)\]
In other words, perturbing the states only perturbs the joint probability matrices by a small amount, which increases only polynomially in the number of qubits. In contrast, a robust theory is only required to obey $|P(\psi,U)-P(\phi, U')|_1 \leq \poly(\epsilon)\poly(N)$, i.e. the joint probability matrices can be perturbed by an amount which increases polynomially in the dimension of the Hilbert space.

If a strongly robust a theory exists, it would immediately imply a lower bound for search in $\DQP$ which is polynomial in $N$ - the reason is that for this marked item $x$, we would have 
\[|P(\psi,U)-P(\psi^x, U^x)|_1 \leq \poly (\epsilon) \polylog(N) = \poly\left(\frac{Q}{\sqrt{N}}\right)\polylog(N)\]
at all stages of the algorithm, and hence by Lemma ~\ref{union},
\begin{align*} d_{TV}(v, v^x) &\leq \displaystyle\sum_i d_{TV}((v_{i-1}, v_i), (v^x_{i-1}, v^x_i)) 
\\ &= \displaystyle\sum_t |P(\psi^x_t, U^x_t)-P(\psi_t, U_t)|_1
\\ &\leq T\poly\left(\frac{Q}{\sqrt{N}}\right)\polylog(N)
\end{align*}
Since the $\DQP$ search algorithm must work for this strongly robust theory, we must have $\frac{TQ^c\polylog(N)}{N^{c/2}} \geq d_{TV}(v,v^x) \geq \Omega (1)$ for some constant $c$ which is the exponent of the polynomial in $\epsilon$. This implies $T+Q=\tilde \Omega (N^{c/(2+2c)})$. Note that perturbing a state by $\epsilon$ has to perturb the resulting $P$ matrices by at least $\epsilon$ (since it must alter their row sums by $\epsilon$), and hence we must have $0<c\leq 1$. Therefore even if a strongly robust theory exists, the best possible lower bound one could prove using this technique is $N^{1/4}$.

Unfortunately we do not know of any theories which are strongly robust. The only provably robust theory we know of is flow theory, which in \cite{hidden} is shown to obey
\[|P(\psi,U)-P(\psi^x, U^x)|_1 \leq 4\epsilon N^2\]
which does not meet the criteria for strong robustness. An interesting open problem is to determine if flow theory, Sch\"{o}dinger theory (described in \cite{hidden}), or any hidden variable theory is strongly robust.

\section{An $N^{1/4}$ lower bound for search in a modified version of \textsf{DQP}}\label{app:dqpfix}

Although we do not know how to prove a polynomial lower bound for search in $\DQP$, we can show an $N^{1/4}$ lower bound for search in a modified version of $\DQP$, which we describe below:

We first modify the definition of a hidden variable theory. A hidden variable theory is a function $P(\psi, C)$ which depends on 
\begin{enumerate}
\item A quantum state $\psi = \Sigma_i \alpha_i \ket{i}$
\item A quantum circuit $C$ which specifies product of unitary gate elements $g^k$, $k=1\ldots poly(n)$, from some universal gate set $\mathcal{U}$. Note $U=\Pi_k g^k$.
\end{enumerate}

Unlike before, we now allow $P(\psi,C)$ to depend on the circuit generating the unitary $U$, rather than only the unitary itself. The output of $P(\psi, C)$ is a joint probability matrix $P_{ij}$, $i,j=1\ldots N$ which satisfies
\begin{enumerate}
\item $\Sigma_j P_{ij} = |\alpha_i|^2$ where $\psi = \Sigma_i \alpha_i \ket{i}$ 
\item $\Sigma_i P_{ij} = |\beta_j|^2$ where we have $U\psi = \Sigma_j \beta_j \ket{j}$
\end{enumerate}
as before.

We call $B\subseteq [N]$ a \emph{circuit block} for circuit $C = \Pi_k g^k$, where each $g^k$ is a gate from a universal gate set $\mathcal{U}$, if for all $k$, $g^k_{ij}=0$ for all $i\in B, j\notin B$ and $g^k_{ij}=0$ for all $i\notin B$, $j\in B$. In other words, a circuit block $B$ is valid if for all circuit elements $g_k$, indices $i,j$ are in the same block in the unitary $g_k$. The \emph{circuit block structure} of $C$ is a minimal collection of circuit blocks which partition $[N]$. 

In contrast, the block structure of $C$ is the block structure of the resulting unitary. Note that block structure of $C$ is always a refinement of its circuit block structure; if all gates in $C$ have $B$ as a valid block, then the final unitary will have $B$ as a valid block, but the converse is not true. For example, suppose that $C=HH$ on a single qubit. Since $U=HH=I$ the block structure of $C$ is $\{1\},\{2\}$. However the circuit block structure of $C$ is $\{1,2\}$, i.e. the trivial circuit block structure, because the individual circuit elements do not have any block structure.

We call a hidden variable theory \emph{circuit-indifferent} if $P(\psi, C)$'s block structure respects the circuit block structure of $C$. Since the block structure of a unitary $U$ is always a refinement of the circuit block structure of the circuit $C$ producing $U$, an indifferent theory is always circuit-indifferent. Hence the set of circuit-indifferent theories is larger than the set of indifferent theories.

We define a new version of \textsf{DQP}, which we call $\CDQP$ (for ``circuit-indifferent $\DQP$"), as before, except 
\begin{enumerate}
\item We require the algorithms to work for all circuit-indifferent hidden variable theories
\item We no longer require the hidden variable theories to be robust. As a result the definition of our class is gate set dependent. Assume we have all 1 and 2-qubit gates at our disposal.
\item When given access to a search oracle $f:\{0,1\}^n\rightarrow \{0,1\}$, we assume it is a phase oracle, i.e. $\mathcal{O}_f \ket{x}=(-1)^{f(x)}\ket{i}$. This distinction did not matter in the definition of $\DQP$ or $\PDQP$, but it does matter here, because our hidden variable theories depend on the block structure of individual circuit elements, including the oracle.
\end{enumerate}

We can now prove a lower bound for search in this version of $\CDQP$. 

\begin{theorem}
Any algorithm correctly deciding search in $\CDQP$ using $Q$ queries and $T$ time satisfies  $Q+T=\Omega(N^{1/4})$.
\end{theorem}
\begin{proof}

We will describe a circuit-indifferent hidden variable theory, which we call Dieks theory for circuit block structure, which foils any search algorithm $A$ which uses $Q+T = o(N^{1/4})$ time. This contradicts the requirement that $A$ work for all circuit-indifferent hidden variable theories.

Suppose that $A$ generates quantum circuits $C_1 \ldots C_T$ when there is no marked item, and quantum circuits $C_1^x \ldots C_T^x$ when there is a marked item at location $x$. Clearly the circuits $C_t$ and $C_t^x$ differ only in their search oracles. The search oracles are diagonal,  hence $C_t$ and $C_t^x$ have the same circuit block structure $I$. This will be crucial in proving our result.

Let $\psi_t$ be the quantum state after $t$ steps of the algorithm when there is no marked item, and let $\psi_t^x$ be the quantum state after $t$ steps when there is a marked item at location $x$. By the hybrid argument, there exists an item $x$ such that
\begin{align}||\psi_t - \psi_t^x|| \leq \frac{4Q}{\sqrt{N}} \label{eq:hybridcdqp}\end{align}
for all $t=1\ldots T$, where $||\psi_t - \psi_t^x||$ indicates the trace norm.

We will show that if $P(\psi_t, C_{t+1}, t)$ and $P(\psi_t^x, C_{t+1}^x)$ are given by Dieks theory for circuit block structure, then
\begin{align} |P(\psi_t, C_{t+1}) - P(\psi_t^x, C_{t+1}^x)|_1 \leq \frac{12Q}{\sqrt{N}} \label{eq:hybridcdqp2}\end{align}
From this the lower bound will follow, because the trace distance between the hidden variable histories with and without a marked item is upper bounded by
\[\displaystyle\sum_t |P(\psi_t, C_{t+1}) - P(\psi_t^x, C_{t+1}^x)|_1 \leq O\left(\frac{TQ}{\sqrt{N}}\right) \]
by Lemma \ref{union}. The quantity must be $\Omega(1)$ because $A$ distinguishes the presense of a marked item with $\Omega(1)$ probability. Hence we have $TQ=\Omega(N^{1/2})$ so $T+Q=\Omega(N^{1/4})$ as desired.

We now define Dieks theory for circuit block structure. Let $I$ be the circuit block structure of $C$. Let $P:=P(\psi_t, C_{t+1})$ be the joint probability matrix of Dieks theory with block structure $I$. That is, 
\[
P_{ij}= |\alpha_i|^2 \frac{|\beta_j|^2}{\Sigma_{\hat{j}\in B} |\beta_{\hat{j}}|^2}
\]
if $i,j$ are in the same block $B\in I$ and $0$ otherwise. Note $P$ is a valid, circuit indifferent matrix. Indeed the column and row sums are
\begin{align}
\displaystyle\sum_j P_{ij} &= |\alpha_i|^2 \displaystyle\sum_{j\in B} \frac{|\beta_j|^2}{\Sigma_{\hat{j}\in B} |\beta_{\hat{j}}|^2} = |\alpha_i|^2\\ 
\displaystyle\sum_i P_{ij} &= \displaystyle\sum_{i\in B} |\alpha_i|^2 \frac{|\beta_j|^2}{\Sigma_{\hat{j}\in B} |\beta_{\hat{j}}|^2} \\
 &=  |\beta_j|^2\frac{\Sigma_{i\in B} |\alpha_i|^2}{\Sigma_{\hat{j}\in B} |\beta_{\hat{j}}|^2} = |\beta_j^2| \label{eq:dieksvalid}
\end{align}
where in line \ref{eq:dieksvalid} we used the fact that the actual block structure of $U$ is a refinement of the circuit block structure of $C$, hence $U$ restricted to any block $B$ of $I$ is also unitary, and so $\Sigma_{i\in B} |\alpha_i|^2 = \Sigma_{\hat{j}\in B} |\beta_{\hat{j}}|^2$. Hence $P(\psi,C)$ is a valid circuit-indifferent hidden variable theory. 

The following Lemma, combined with the equation \ref{eq:hybridcdqp} and the fact that $C$ and $C^x$ have the same circuit block structure, implies equation \ref{eq:hybridcdqp2}.

\begin{lemma}
Suppose that $||\psi-\psi_x||\leq ||U\psi - U^x \psi^x||\leq \epsilon$ where $U$ ($U^x$) is the unitary produced by circuit $C$ ($C^x$). Furthermore suppose $C$ and $C^x$ have the same circuit block structure. Then if $P$ is given by Dieks theory for circuit block structure, then $|P(\psi,C) - P(\psi^x,C^x)|\leq 3\epsilon$.
\end{lemma}
\begin{proof}

Let $\alpha_i,\alpha_i^x, \beta_i, \beta_i^x$ be defined by $\psi=\Sigma_i \alpha_i \ket{i}$, $\psi^x=\Sigma_i \alpha_i^x \ket{i}$, $U\psi=\Sigma_i \beta_i \ket{i}$, and $U\psi^x=\Sigma_i \beta_i^x \ket{i}$ as usual.

Let $I$ be the circuit block structure of $C$ and $C^x$. By the definition of Dieks theory for circuit indifference, we have that $P:=P(\psi, C)$  and $\hat{P}:=P(\psi^x, C^x)$ are given by
\begin{align*}
P_{ij} = \begin{cases}
|\alpha_i|^2 \frac{|\beta_j|^2}{\Sigma_{\hat{j}\in B} |\beta_{\hat{j}}|^2} & i,j\in B \in I\\
0 & \text{o.w.}
\end{cases}
& & 
\hat{P}_{ij} = \begin{cases}
|\alpha^x_i|^2 \frac{|\beta^x_j|^2}{\Sigma_{\hat{j}\in B} |\beta^x_{\hat{j}}|^2} & i,j\in B \in I\\
0 & \text{o.w.}
\end{cases}
\end{align*}
We can now show $\hat{P}$ is close to $P$ in trace distance. Note that
\begin{align}
|P-\hat{P}|_1 &= \displaystyle\sum_{i,j} \left| |\alpha_i|^2 \frac{|\beta_j|^2}{\Sigma_{\hat{j}\in B} |\beta_{\hat{j}}|^2} - |\alpha^x_i|^2 \frac{|\beta^x_j|^2}{\Sigma_{\hat{j}\in B} |\beta^x_{\hat{j}}|^2}\right| \label{eq:step1} \\
&\leq \displaystyle\sum_B \displaystyle\sum_{i,j\in B} \left| |\alpha_i|^2 \frac{|\beta_j|^2}{\Sigma_{\hat{j}\in B} |\beta_{\hat{j}}|^2} - |\alpha_i|^2 \frac{|\beta^x_j|^2}{\Sigma_{\hat{j}\in B} |\beta^x_{\hat{j}}|^2}\right| + \left|  |\alpha_i|^2 \frac{|\beta^x_j|^2}{\Sigma_{\hat{j}\in B} |\beta^x_{\hat{j}}|^2} - |\alpha^x_i|^2 \frac{|\beta^x_j|^2}{\Sigma_{\hat{j}\in B} |\beta^x_{\hat{j}}|^2}\right| \label{eq:step2}\\
&= \displaystyle\sum_B \displaystyle\sum_{i, j\in B} |\alpha_i|^2 \left|  \frac{|\beta_j|^2}{\Sigma_{\hat{j}\in B} |\beta_{\hat{j}}|^2} - \frac{|\beta^x_j|^2}{\Sigma_{\hat{j}\in B} |\beta^x_{\hat{j}}|^2}\right| + \displaystyle\sum_B \displaystyle\sum_{i,j\in B} \frac{|\beta^x_j|^2}{\Sigma_{\hat{j}\in B} \Sigma_{i\in B} |\beta^x_{\hat{j}}|^2} \left||\alpha_i|^2- |\alpha^x_i|^2\right| \label{eq:step3}\\
&= \displaystyle\sum_B \displaystyle\sum_{j\in B}\displaystyle\sum_{\hat{j}\in B} |\beta_{\hat{j}}|^2 \left|  \frac{|\beta_j|^2}{\Sigma_{\hat{j}\in B} |\beta_{\hat{j}}|^2} - \frac{|\beta^x_j|^2}{\Sigma_{\hat{j}\in B} |\beta^x_{\hat{j}}|^2}\right| + \displaystyle\sum_i \left||\alpha_i|^2- |\alpha^x_i|^2\right| \label{eq:step4}\\
&\leq \displaystyle\sum_B \displaystyle\sum_{j\in B} \left|  |\beta_j|^2 - |\beta^x_j|^2\frac{\Sigma_{\hat{j}\in B} |\beta_{\hat{j}}|^2}{\Sigma_{\hat{j}\in B} |\beta^x_{\hat{j}}|^2}\right| + \epsilon \label{eq:step5} \\
& \leq \displaystyle\sum_B \displaystyle\sum_{j\in B} \left|  |\beta_j|^2 - |\beta_j^x|^2\right| +\left|  |\beta_j^x|^2 - |\beta^x_j|^2\frac{\Sigma_{\hat{j}\in B} |\beta_{\hat{j}}|^2}{\Sigma_{\hat{j}\in B} |\beta^x_{\hat{j}}|^2}\right|+ \epsilon \label{eq:step6} \\
& \leq \epsilon + \displaystyle\sum_B \displaystyle\sum_{j\in B}  |\beta_j^x|^2 \left| 1 - \frac{\Sigma_{\hat{j}\in B} |\beta_{\hat{j}}|^2}{\Sigma_{\hat{j}\in B} |\beta^x_{\hat{j}}|^2}\right|+\epsilon \label{eq:step7} \\
& = \epsilon + \displaystyle\sum_B \left| \displaystyle\sum_{j\in B}  |\beta_j^x|^2 - \displaystyle\sum_{\hat{j}\in B} |\beta_{\hat{j}}|^2 \right|+\epsilon \label{eq:step8} \\
&\leq 3\epsilon \label{eq:step9} 
\end{align}
where line (\ref{eq:step2}) follows from the triangle inequality, line (\ref{eq:step4}) from the fact that $U$ has block structure $I$ so $\Sigma_{i\in B} |\alpha_i|^2 = \Sigma_{j\in B} |\beta_j|^2$ as well as an evalution of the second sum, line (\ref{eq:step5}) from our upper bound on the trace distance of $\psi$ and $\psi_x$, line (\ref{eq:step6}) by the triangle inequality, and lines (\ref{eq:step7}) and (\ref{eq:step9}) by our upper bound on the trace distance of $U\psi$ and $U\psi^x$. This completes the proof.

\end{proof}

Hence Dieks theory for circuit block structure foils any $\CDQP$ algorithm taking less than $N^{1/4}$ time, which completes the proof.

\end{proof}

\section{Strange properties of non-collapsing measurements}\label{bizarreqm}

Here we show why allowing non-collapsing measurements in quantum mechanics allows for faster than light communication, allows for quantum cloning, and renders quantum query complexity and quantum communication complexity meaningless.

To see that non-collapsing measurements allow for faster-than-light communication: suppose two players share an EPR pair, and one player makes a collapsing measurement either in the 0/1 basis or in the +/-  basis. By performing non-collapsing measurements on their half of the state, the second player can tell (with high probability) which basis the first player measured in, and hence receive a signal faster than light.

To see that non-collapsing measurements allow for cloning: given a quantum state $\psi$ on $n$ qubits, one could perform $2^{O(n)}$ non-collapsing measurements to characterize the state using tomography, and then (approximately) reproduce the state. This ``approximate cloning" operation take exponential time for generic states, and is non-unitary. Note, that the class of computations considered in $\PDQP$ cannot perform this operation, even on states of $O(\log(n))$ qubits, since the $\PDQP$ machine cannot perform further quantum computations after receiveing the non-collapsing measurement results. In other words, the quantum circuit in $\PDQP$ cannot depend on the non-collapsing measurement outcomes. An interesting open problem is whether or not allowing the quantum circuit to depend on the non-collapsing measurement results changes the power of the class. Or, more generally, what is the power of quantum computers which are given the ability to clone?

We now explain why with non-collapsing measurements, the quantum query complexity and quantum communication complexity of any function is $1$. Suppose one wishes to evaluate $f(x)$ where $x=x_1 \ldots x_N$. Then one can prepare the superposition $\sum_i \ket{i}\ket{x_i}$ with one query to the oracle, and make $O(N\log N)$ non-collapsing measurements of this state to observe the value of each $x_i$ and compute the function. Similarly, in the context of communication complexity, one player can simply encode their input $x\in\{0,1\}^n$ into the state $\cos \theta_x \ket{0} + \sin \theta_x \ket{1}$ where $\theta_x = \frac{x}{2^n}\frac{\pi}{2}$. By performing roughly $2^n$ non-collapsing measurements, the other player can learn $\theta_x$ and hence $x$, with only one quantum bit of communication.  Note that although these example algorithms use only one query or one qubit of communication, respectively, they use a large number of non-collapsing measurements. For this reason, when we prove lower bounds for $\PDQP$, we lower bound the number of queries \emph{plus} the number of non-collapsing measurements required, rather than the number of queries alone.

\section{Universal gate set does not matter} \label{gateset}

\par We prove that the universal gate set $\mathcal{U}$ used in the definition of $\PDQP$ does not matter. Our proof relies on Lemma~\ref{union} and the Solovay-Kitaev theorem~\cite{solovay} to show that any computation using a particular universal gate set $\mathcal{U}$ can be done using a different gate set $\mathcal{U}'$ in such a way that the distributions of the histories does not change significantly in total variation distance.

To do so, we will first give an alternative definition of $\PDQP$ which will make the proof easier. Our alternative definition is framed in the notation of $\DQP$; for an introduction to this notation please see Appendix \ref{errors}.

\subsection{An alternative definition of $\PDQP$}
\par If $B$ is a partition of $\{0,1\}^\ell$ and $U$ is a unitary operator on $(\mathbb{C}^2)^{\otimes \ell}$, then we say that $U$ respects the block structure $B$ if $U_{ij} = 0$ whenever $i$ and $j$ are in different parts of $B$. If $\ket{\psi}$ is a pure state and $U$ is a unitary that respects the block structure $B$, then the stochastic matrix $S_{\mathcal{PT}_B}(\ket{\psi}, U)$ is formed by applying the ``product theory" $\mathcal{PT}$ separately on each block of $B$. More precisely, let $\sim$ be the equivalence relation on $\{1, \cdots, n\}$ defined by $i \sim j$ if and only if $i$ and $j$ are in the same block of $B$. Let $\ket{\psi}=\Sigma_i \alpha_i\ket{i}$ and $U\ket{\psi}=\Sigma_j \beta_j \ket{j}$. Then, 
\[
(S_{\mathcal{PT}_B}(\ket{\psi}, U))_{ij} = \begin{cases}
\frac{|\beta_j|^2}{\sum_{k \sim j} |\beta_{k}|^2} & \mbox{if $i \sim j$} \\
0 & \mbox{otherwise}
\end{cases}\] where the sum over $k$ ranges over all $k$ with $k \sim j$. 
\par Suppose that $\mathcal{V} = (U_1, \cdots, U_T)$ are unitary operators on $\ell$ qubits, and $\mathcal{B} = (B_1, \cdots, B_T)$ are partitions of $\{0, \cdots, 1\}^n$ such that for every $i$, $B_{i+1}$ is a refinement of $B_i$, and $U_i$ respects the block structure $B_i$. Then they define a probability distribution $\Omega = \Omega_{\mathcal{PT}}(\mathcal{V}, \mathcal{B})$ over hidden variable histories $v = (v_0, \cdots, v_T)$ by \[\Omega_{(v_0, \cdots, v_T)} = \prod_{k=1}^{T} (S_{\mathcal{PT}_{B_{k}}}(U_{k-1} \cdots U_1 \ket{0}^{\otimes \ell}, U_{k}))_{v_{k-1} v_{k}}.\] The oracle $\mathcal{Q}_B$ takes as input the unitaries $U_1, \cdots, U_T$ specified by sequences of gates from some finite universal gate set $\mathcal{U}$. It also takes as input the partitions $B_1, \cdots, B_T$, specified by polynomial-time computable functions $b_1, \cdots, b_T : \{0, 1\}^\ell \to \{0, 1\}^m$ satisfying the property that $x$ and $y$ are in the same part of the partition $B_i$ if and only if $b_i(x) = b_i(y)$. It outputs a sample from the distribution $\Omega_{\mathcal{PT}}(\mathcal{V}, \mathcal{B})$. Then, let $\PDQP'$ be the class of all languages that can be recognized by a polynomial-time Turing machine with one query to $\mathcal{Q}_B$, with error probability at most $\frac{1}{3}$.
\begin{lemma}
\label{alternative}
 $\PDQP' = \PDQP$.
\end{lemma}
\begin{proof}
\par We first demonstrate a procedure for converting oracle queries to $\mathcal{Q}_B$ to oracle queries to $\mathcal{Q}_P$. Suppose that $B_1, \cdots, B_T$ are specified by polynomial-time computable functions $b_1, \cdots, b_T : \{0, 1\}^\ell \to \{0, 1\}^m$ (so that $x, y$ are in the same part of the partition $B_i$ if and only if $b_i(x) = b_i(y)$). Now, add an extra $T$ registers of $m$ qubits each, which start in the state $\ket{0\cdots0}$. Create a quantum circuit with the same unitary operators $U_1, \cdots, U_T$, but before applying the unitary $U_i$, apply a unitary that writes the value $\ket{b_i(x)}$ to the $i$th register when the first $\ell$ qubits are $\ket{x}$. Then measure the $i$th register. The effect is that the non-collapsing measurement results will never jump from one part of $B_i$ to a different part, which is exactly what is desired.
\par To convert a query $C = (U_1, M_1, \cdots, U_T, M_T)$ to $\mathcal{Q}_P$ to a query to $\mathcal{Q}_B$, we first assume, as in the proof of Theorem~\ref{lowerbound}, that measured qubits are never modified again. Keep the unitaries $U_1, \cdots, U_T$ and let $B_i$ be the partition of $\{0,1\}^\ell$ induced by the measurements $M_1, \cdots, M_{i-1}$. By the principle of deferred measurement, $\Omega_{\mathcal{V}, \mathcal{B}}$ is the same distribution that we would have seen had we queried $\mathcal{Q}_P$ instead.
\end{proof}

Now that we have given an alternative definition of $\PDQP$, we can easily show that the choice of gate set does not matter:

\begin{theorem}
 Any universal gate set $\mathcal{U}$ yields the same complexity class $\PDQP$.
\end{theorem}
\begin{proof}
\par If $A$ is an operator, denote by $\norm{A}$ the maximum value of $\norm{A \ket{\phi}}_2$ over all $\phi$ with $\norm{\ket{\phi}}_2 = 1$.
\begin{lemma}
Suppose that $V_1, \cdots, V_m$ and $V_1', \cdots, V_m'$ are unitary operators. Then,
\[\norm{V_1 \cdots V_m - V_1' \cdots V_m'} \leq \sum_{k=1}^m \norm{V_k - V_k'}.\]
\end{lemma}
\begin{proof}
By induction, it suffices to prove the statement for $m = 2$. We have
\begin{align*}
 \norm{V_1V_2 - V_1' V_2'} &= \max_{\norm{\ket{\phi}}_2 = 1} \norm{V_1V_2 \ket{\phi}- V_1' V_2'\ket{\phi}}_2 \\
 &\leq \max_{\norm{\ket{\phi}}_2 = 1} (\norm{V_1V_2 \ket{\phi}- V_1' V_2 \ket{\phi}}_2 + \norm{V_1'V_2 \ket{\phi}- V_1 V_2 \ket{\phi}})\\
 &=\max_{\norm{\ket{\phi}}_2 = 1} (\norm{(V_1 - V_1')V_2\ket{\phi}}_2 + \norm{(V_2 - V_2') \ket{\phi}}_2) \\
 &\leq \norm{V_1 - V_1'} + \norm{V_2 - V_2'}.
\end{align*}
\end{proof}
\par If $\ket{\psi}=\Sigma_i \alpha_i \ket{i}$ is a pure state and $U$ is a unitary operator on $\ell$ qubits that respects the block structure $B$, such that $U\ket{\psi}=\Sigma_j \beta_j \ket{j}$, then define the joint probabilities matrix $P_{\mathcal{PT}_B}(\ket{\psi}, U)$ by
 \[
(P_{\mathcal{PT}_B}(\ket{\psi}, U))_{ij} = \begin{cases}
\frac{|\alpha_i|^2 |\beta_{j}|^2}{\sum_{k \sim j} |\beta_{k}|^2} & \mbox{if $i \sim j$} \\
0 & \mbox{otherwise}
\end{cases}.\] It is straightforward to show that \[\norm{P_{\mathcal{PT}_B}(\ket{\psi}, U) - P_{\mathcal{PT}_B}(\ket{\psi'}, U')}_1 \leq 2^{2\ell}(\norm{\ket{\psi} - \ket{\psi'}}_{tr} + \norm{U - U'})\] whenever $\ket{\psi}, \ket{\psi'}$ are state vectors and $U, U'$ are unitary operators. 
\par We use the alternative formulation $\PDQP'$ (Lemma~\ref{alternative}). Suppose that $\mathcal{U}$ and $\mathcal{U}'$ are two universal gate sets, and that $\mathcal{V} = (U_1, \cdots, U_T)$ and $\mathcal{B} = (B_1, \cdots, B_T)$ are a query to the $\mathcal{Q}_B$ oracle, where the operators $U_t$ are specified by sequences of gates from $\mathcal{U}$. It is enough to be able to compute in polynomial time a sequence $\mathcal{V'} = (U_1', \cdots, U_T')$ of unitary operators, specified by sequences of gates from $\mathcal{U}'$, such that \[d_{TV}(\Omega_{\mathcal{PT}}(\mathcal{V}, \mathcal{B}), \Omega_{\mathcal{PT}}(\mathcal{V}', \mathcal{B})) < \frac{1}{8}.\]
\par Let $\epsilon = 2^{-\ell^2 T - 10}$. Then, by the Solovay-Kitaev theorem~\cite{solovay}, it is possible to compute in polynomial time a sequence $\mathcal{V}' = (U_1', \cdots, U_T')$ such that \[\norm{U_t - U_t'} \leq \epsilon\] for all $t$. Suppose that $v = (v_0, \cdots, v_T)$ is sampled from $\Omega_{\mathcal{PT}}(\mathcal{V}, \mathcal{B})$, and that $v' = (v_0', \cdots, v_T')$ is sampled from $\Omega_{\mathcal{PT}}(\mathcal{V}', \mathcal{B})$. Then, \[d_{TV}(\Omega_{\mathcal{PT}}(\mathcal{V}, \mathcal{B}), \Omega_{\mathcal{PT}}(\mathcal{V}', \mathcal{B})) = d_{TV}(v, v').\] By Lemma~\ref{union}, 
\begin{align*}
 d_{TV}(v, v') &\leq 2 \sum_{i=1}^{T} d_{TV}((v_{i-1}, v_i), (v_{i-1}', v_i')) \\
 &= 2 \sum_{i=1}^T \norm{P_{\mathcal{PT}_{B_i}}(U_{i-1} \cdots U_1\ket{0}^{\otimes \ell}, U_i) - P_{\mathcal{PT}_{B_i}}(U_{i-1}' \cdots U_1'\ket{0}^{\otimes \ell}, U_i')}_1 \\
 &\leq 2^{2\ell + 1} \sum_{i=1}^T\left( \norm{U_{i-1} \cdots U_1\ket{0}^{\otimes \ell} - U_{i-1}' \cdots U_1'\ket{0}^{\otimes \ell}}_2 + \norm{U_i - U_i'}\right) \\
 &\leq 2^{2\ell + 1} \sum_{i=1}^T\left(\norm{U_{i-1} \cdots U_1 - U_{i-1}' \cdots U_1'}+ \epsilon\right) \\
 &\leq 2^{2\ell + 1} \sum_{i=1}^T\left(\sum_{k=0}^{i-1} \norm{U_i - U_i'}+ \epsilon\right) \\
 &\leq 2^{2\ell + 1} \sum_{i=1}^T\left(T \epsilon+ \epsilon\right) \\
 &\leq \frac{1}{8},
\end{align*}
as desired.
\end{proof}

\section{An $N^{1/3}$ lower bound for search in $\PDQP$ if there are no collapsing measurements}\label{app:search13}

Assume that intermediate measurements are not allowed in our search algorithm. As we said before, this gives a model with only the power of $\BQP$, because then the states $\ket{\psi_t}=U_t U_{t-1} \ldots U_1 \ket{0}^{\otimes n}$ can be generated with poly-sized circuits, and hence a $\BQP$ machine could prepare and and measure them to sample from $\mathcal{Q}_P$. Trivially one can prove a lower bound of $N^{1/4}$ for search in this model, either by noting that this class can achieve at most quadratic speedups over $\BQP$ by the previous comment, or by using the argument put forth in Theorem \ref{lowerbound}. Here we tighten this result to give an $N^{1/3}$ lower bound for search in this class.

Suppose that an algorithm $A$ searches with $Q$ queries and $T$ timesteps, where $Q+T=o(N^{1/3})$. Let $\psi_t$ be the quantum state after $t$ steps with no marked item, and let $\psi_t^x$ be defined likewise when the marked item is at location $x$. By the hybrid argument we have that $\forall t$
\[ \displaystyle\sum_x||\psi_t - \psi_t^x||_2^2 \leq 4Q^2\]
where $||a||_2^2$ is the 2-norm squared of $a$. This implies
\[\displaystyle\sum_t \displaystyle\sum_x||\psi_t - \psi_t^x||_2^2 \leq 4TQ^2\]
Hence there must exist $x$ such that
\begin{equation}\displaystyle\sum_t ||\psi_t - \psi_t^{x}||_2^2 \leq \frac{4TQ^2}{N}\label{eq:productx}\end{equation}
Since we assumed $Q+T=o(N^{1/3})$, we have that $\frac{4TQ^2}{N}=o(1)$. Therefore for sufficiently large $N$ and for all $t$ we have
\[||\psi_t - \psi_t^{x}||_2^2 \leq 0.01\]

(The choice of constant here is arbitrary, we simply need it to be less than around 0.5.)

Now consider the states $\Psi :=\bigotimes_t \ket{\psi_t}$ and $\Psi^{x} :=\bigotimes_t \ket{\psi^{x}_t}$. Let $V$ the distribution on samples with no marked item, and let $V^x$ be defined likewise. Then clearly we have that
\[|V-V^x|_1 \leq  ||\Psi - \Psi^x||\]
where $||a||$ denotes the trace norm of $a$. This is because the output distributions of $V$ and $V^x$ can be obtained by (independent) measurements on the states $\Psi$ and $\Psi_x$ in the computational basis. Note that $|V-V^x|_1$ must be $\Omega(1)$ in order to distinguish the presence of a marked item at $x$ in postprocessing. Therefore we have
\begin{align}
\Omega(1) \leq |V-V_x|_1 &\leq ||\Psi - \Psi^x||\\
&= \sqrt{1-|\braket{\Psi}{\Psi^x}|^2} \label{eq:product1}\\
&= \sqrt{1-|\Pi_t \braket{\psi_t}{\psi_t^x}|^2} \\
&\leq \sqrt{1-\Pi_t e^{-||\psi_t - \psi_t^x||_2^2}} \label{eq:product2} \\
&= \sqrt{1 - e^{-\Sigma_t ||\psi_t - \psi_t^x||_2^2}} \\
&\leq \sqrt{1 - e^{-\frac{4TQ^2}{N}}} \label{eq:product3}\\
&= o(1) \label{eq:product4}
\end{align}

Where in line \ref{eq:product1} we use the formula for trace distance of pure states, in line \ref{eq:product3} we used equation \ref{eq:productx}, in line \ref{eq:product4} we used the fact that $T+Q=o(N^{1/3})$, and in line \ref{eq:product2} we use the inequality 
\begin{align}
|\braket{\psi_t}{\psi_t^x}| &\geq \mathrm{Re}\left(\braket{\psi_t}{\psi_t^x}\right) \\
&=1-\frac{||\psi_t - \psi_t^x||_2^2}{2} \\
&\geq e^{-||\psi_t - \psi_t^x||_2^2}
\end{align}
where we have use the fact that $1-x \geq e^{-2x}$ for $0\leq x \leq 0.01$. 

Therefore we have shown $\Omega(1)=o(1)$, a contradiction. Hence such an algorithm $A$ cannot exist, so searching takes $Q+T=\Omega(N^{1/3})$ time when there are non-collapsing measurements, but no collapsing measurements, in the model.

\end{appendices}

\end{document}